\documentclass[12pt,letterpaper]{article}
\usepackage{amsmath,amsfonts,amsthm,amssymb}

\theoremstyle{plain}

\numberwithin{equation}{section}
\newtheorem{thm}{Theorem}[section]
\newtheorem{lem}[thm]{Lemma}
\newtheorem{cor}[thm]{Corollary}

\newenvironment{exam}[1]{\medskip%
  \setlength{\rightmargin}{\leftmargin}%
  {\noindent\textbf{Example #1}.\enspace}%
  }%

\newcounter{cond}

\newcommand{\complex}{{\mathbb C}}
\newcommand{\real}{{\mathbb R}}
\newcommand{\integers}{{\mathbb Z}}
\newcommand{\ascript}{{\mathcal A}}
\newcommand{\bscript}{{\mathcal B}}
\newcommand{\cscript}{{\mathcal C}}
\newcommand{\mscript}{{\mathcal M}}
\newcommand{\pscript}{{\mathcal P}}
\newcommand{\sscript}{{\mathcal S}}
\newcommand{\deltahat}{\widehat{\delta}}
\newcommand{\deltahatij}{\deltahat _{ij}}
\newcommand{\muhat}{\widehat{\mu}}
\newcommand{\rmre}{\mathrm{Re\,}}
\newcommand{\rmext}{\mathrm{Ext\,}}

\newcommand{\cupdot}{\mathbin{\cup{\hskip-5.4pt}^\centerdot}\,}
\newcommand{\bigcupdotionen}{{\bigcup _{i=1}^n}{\hskip-8pt}^\centerdot{\hskip 8pt}}
\newcommand{\bigcupdotitwon}{{\bigcup _{i=2}^n}{\hskip-8pt}^\centerdot{\hskip 8pt}}

\newcommand{\ab}[1]{\left|#1\right|}
\newcommand{\brac}[1]{\left\{#1\right\}}
\newcommand{\paren}[1]{\left(#1\right)}
\newcommand{\sqbrac}[1]{\left[#1\right]}

\begin{document}

\title{QUANTUM MEASURES AND\\THE COEVENT INTERPRETATION}
\author{Stan Gudder\\ Department of Mathematics\\
University of Denver\\ Denver, Colorado 80208\\
sgudder@math.du.edu}
\date{}
\maketitle

\begin{abstract}
This paper first reviews quantum measure and integration theory. A new representation of the quantum integral is presented. This representation is illustrated by computing some quantum (Lebesgue)${}^2$ integrals. The rest of the paper only considers finite spaces. Anhomomorphic logics are discussed and the classical domain of a coevent is studied. Pure quantum measures and coevents are considered and it is shown that pure quantum measures are precisely the extremal elements for the set of quantum measures bounded above by one. This result is then used to prove a connection between quantum measures on a (finite) event space $\ascript$ and an anhomomorphic logic $\ascript ^*$. To be precise, it is shown that any quantum measure on
$\ascript$ can be transferred to an ordinary measure on $\ascript ^*$. In this way, the quantum dynamics on $\ascript$ can be described by a classical dynamics on the larger space $\ascript ^*$. 
\end{abstract}

\section{Introduction}  
Quantum measures and the coevent interpretation were introduced by R. Sorkin in his studies of the histories approach to quantum mechanics and quantum gravity
\cite{hal09, mocs05, sor941, sor942, sor07}. Since then a considerable amount of literature has appeared on these subjects \cite{dgt08, gt09, gtw091, gtw092, gudamm, gudms, gud091, gud092, gud093, gud10, sal02, sor09, sw08}. A quantum measure $\mu$ describes the dyamics of a quantum system in the sense that $\mu (A)$ gives the propensity that the event $A$ occurs. Denoting the set of events by $\ascript$, a coevent is a potential reality for the system given by a truth function $\phi\colon\ascript\to\integers _2$ where $\integers _2$ is the two element Boolean algebra $\brac{0,1}$. Since coevents need not be Boolean homomorphisms, the set $\ascript ^*$ of coevents is called an \textit{anhomomorphic logic}
\cite{gt09, gtw091, gtw092, gud092, gud093, gud10, sor941, sor09}. We refer to the study of anhomomorphic logics as the coevent interpretation of quantum mechanics. One of the goals of this field is to find the ``actual reality'' in $\ascript ^*$ \cite{gt09, gud092, gud10, sor941, sor942, sor09}.

In Section~2, we first review quantum measure and integration theory. A new representation of the quantum integral is presented. As an example, this representation is employed to compute quantum (Lebesgue)${}^2$ integrals. Although general event spaces are treated in Section~2, only finite event spaces are considered in the rest of the paper. Section~3 discusses anhomomorphic logics. The center and classical domain of a coevent are studied. Section~4 first considers pure quantum measures and coevents. We show that the pure quantum measures are precisely the extremal elements for the convex set of quantum measures bounded above by one. This result is then employed to prove an important connection between quantum measures and anhomomorphic logics. More precisely, we show that any quantum measure on a (finite) event space $\ascript$ can be transferred to an ordinary measure on $\ascript ^*$. In this way the quantum dynamics on $\ascript$ can be described by a classical dynamics on the larger space
$\ascript ^*$. Variations of this transference result show that pure and quadratic anhomomorphic logics have better mathematical properties than additive or multiplicative ones.

\section{Quantum Measures and Integrals} 
Let $(\Omega ,\ascript )$ be a measurable space, where $\Omega$ is a set of \textit{outcomes} and $\ascript$ is a $\sigma$-algebra of subsets of $\Omega$ called \textit{events} for a physical system. If $A,B\in\ascript$ are disjoint, we denote their union by $A\cupdot B$. A nonnegative set function $\mu\colon\ascript\to\real ^+$ is \textit{grade}-2 \textit{additive} if
\begin{equation}         
\label{eq21}
\mu\paren{A\cupdot B\cupdot C}=\mu\paren{A\cupdot B}+\mu\paren{A\cupdot C}
+\mu\paren{B\cupdot C}-\mu (A)-\mu (B)-\mu (C)
\end{equation}
for all mutually disjoint $A,B,C\in\ascript$. It follows by induction that if $\mu$ satisfies \eqref{eq21} then
\begin{equation}         
\label{eq22}
\mu\paren{\bigcupdotionen A_i}=\sum _{i<j=1}^n\mu\paren{A_i\cupdot A_j}
-(n-2)\sum _{i=1}^n\mu (A_i)
\end{equation}
A $q$-\textit{measure} is a grade-2 additive set function $\mu\colon\ascript\to\real ^+$ that satisfies the following two continuity conditions.
\begin{list} {(C\arabic{cond})}{\usecounter{cond}
\setlength{\rightmargin}{\leftmargin}}
\item If $A_1\subseteq A_2\subseteq\cdots$ is an increasing sequence in $\ascript$, then
\newline
$\displaystyle{\lim _{n\to\infty}\mu (A_n)=\mu\paren{\bigcup _{i=1}^\infty A_i}}$
\item If $A_1\supseteq A_2\supseteq\cdots$ is a decreasing sequence in $\ascript$, then
\newline
$\displaystyle{\lim _{n\to\infty}\mu (A_n)=\mu\paren{\bigcap _{i=1}^\infty A_i}}$
\end{list}

Due to quantum interference, a $q$-measure need not satisfy the usual additivity condition of an ordinary measure but satisfies the more general grade-2 additivity condition \eqref{eq21} instead
\cite{gt09, sor941, sor942}. For example, if $\nu\colon\ascript\to\complex$ is a complex measure corresponding to a quantum amplitude, then $\mu (A)=\ab{\nu (A)}^2$ becomes a $q$-measure. For another example, let $D\colon\ascript\times\ascript\to\complex$ be a decoherence functional as studied in the histories approach to quantum mechanics \cite{hal09, mocs05, sor941, sor942}. Then $\mu (A)=D(A,A)$ becomes a $q$-measure. If $\mu$ is a $q$-measure on $\ascript$, we call $(\Omega ,\ascript ,\mu )$ a $q$-\textit{measure space}.

A signed measure $\lambda\colon\ascript\times\ascript\to\real$, where $\ascript\times\ascript$ is the product $\sigma$-algebra on $\Omega\times\Omega$, is \textit{symmetric} if
$\lambda (A\times B)=\lambda (B\times A)$ for all $A,B\in\ascript$ and is \textit{diagonally positive} if $\lambda (\ascript\times\ascript )\ge 0$ for all $A\in\ascript$. It can be shown that if
$\lambda\colon\ascript\times\ascript\to\real$ is a symmetric diagonally positive signed measure, then $\mu (A)=\lambda (A\times A)$ is a $q$-measure on $\ascript$. Conversely, if $\mu$ is a
$q$-measure on $\ascript$, then there exists a unique symmetric signed measure
$\lambda\colon\ascript\times\ascript\to\real$ such that $\mu (A)=\lambda (A\times A)$ for all
$A\in\ascript$ \cite{gudms}.

Let $\Omega _n=\brac{\omega _1,\ldots ,\omega _n}$ be a finite outcome space. In this case we take $\ascript _n$ to be the power set $P(\Omega _n)=2^{\Omega _n}$. Any grade-2 additive set function $\mu\colon\ascript _n\to\real ^+$ is a $q$-measure because (C1) and (C2) hold automatically. We call $(\Omega _n,\ascript _n,\mu )$ a \textit{finite} $q$-\textit{measure space}. The study of finite $q$-measure spaces is simplified because by \eqref{eq22} the $q$-measure $\mu$ is determined by its values on singleton sets $\mu (\omega _i)=\mu\paren{\brac{\omega _i}}$, $i=1,\ldots ,n$ and doubleton sets $\mu\paren{\brac{\omega _i,\omega _j}}$, $i<j$,
$i,j=1,\ldots ,n$. We call
\begin{equation}         
\label{eq23}
I_{ij}^\mu =\mu\paren{\brac{\omega _i,\omega _j}}-\mu (\omega _i)-\mu (\omega _j)
\end{equation}
the $ij$-\textit{interference term}, $i<j$, $i,j=1,\ldots ,n$ \cite{gudamm}. The \textit{Dirac measure}
$\delta _i$ defined by
\begin{equation*}
\delta _i(A)=
\begin{cases}
1&\hbox{if }\omega _i\in A\\0&\hbox{if }\omega _i\notin A
\end{cases}
\end{equation*}
is clearly a $q$-measure on $\ascript _n$. The product $\deltahatij =\delta _i\delta _j$, $i\ne j$, is not a measure because $\deltahat _{ij}(\omega _i)=\deltahatij (\omega _j)=0$ and 
$\deltahatij\paren{\brac{\omega _i,\omega _j}}=1$. However, we have the following .
\begin{lem}       
\label{lem21}
The product $\deltahatij$, $i\ne j$, is a $q$-measure on $\ascript _n$.
\end{lem}
\begin{proof}
Notice that
\begin{equation*}
\deltahatij (A)=
\begin{cases}
1&\hbox{if }\brac{\omega _i,\omega _j}\subseteq A\\
  0&\hbox{if }\brac{\omega _i,\omega _j}\not\subseteq A
\end{cases}
\end{equation*}
If $\delta _{ij}(A\cupdot B\cupdot C)=0$, then
$\brac{\omega _i,\omega _j}\not\subseteq A\cupdot B\cupdot C$. Hence,
$\brac{\omega _i,\omega _j}$ is not a subset of $A$, $B$, $C$, $A\cupdot B$, $A\cupdot C$ or
$B\cupdot C$. Hence,
\begin{equation*}
\deltahatij (A\cupdot B)+\deltahatij (A\cupdot C)+\deltahatij (B\cupdot C)
  -\deltahatij (A)-\deltahatij (B)-\deltahatij (C)=0
\end{equation*}
If $\deltahatij (A\cupdot B\cupdot C)=1$, then
$\brac{\omega _i,\omega _j}\subseteq A\cupdot B\cupdot C$. If
$\brac{\omega _i,\omega _j}\subseteq A$, then
\begin{align*}
\deltahatij &(A\cupdot B)+\deltahatij (A\cupdot C)+\deltahatij (B\cupdot C)
  -\deltahatij (A)-\deltahatij (B)-\deltahatij (C)\\
  &=1+1+0-1-0-0=1
\end{align*}
If $\omega _i\in A$, $\omega _j\in B$, then
\begin{align*}
\deltahatij &(A\cupdot B)+\deltahatij (A\cupdot C)+\deltahatij (B\cupdot C)
  -\deltahatij (A)-\deltahatij (B)-\deltahatij (C)\\
  &=1+0+0-0-0-0=1
\end{align*}
The other cases hold by symmetry.
\end{proof}

A more general proof than that in Lemma~\ref{lem21} shows that the product of any two measures is a $q$-measure. A \textit{signed} $q$-\textit{measure} is a set function
$\mu\colon\ascript _n\to\real ^+$ that satisfies \eqref{eq21}. Letting $\sscript (\ascript _n)$ be the set of signed $q$-measures, it is clear that $\sscript (\ascript _n)$ is a real linear space. Also, a signed $q$-measure is determined by its values on singleton and doubleton sets.

\begin{thm}       
\label{thm22}
The $q$-measures
$\delta _1,\ldots ,\delta _n,\deltahat _{12},\deltahat _{13},\ldots ,\deltahat _{n-1,n}$ form a basis for
$\sscript (\ascript _n)$.
\end{thm}
\begin{proof}
To show linear independence, suppose that
\begin{equation*}
\sum _{i=1}^nc_i\delta _i+\sum _{i<j=1}^nd_{ij}\deltahatij =0
\end{equation*}
Evaluating at $\brac{\omega _k}$ gives $c_k=0$, $k=1,\ldots ,n$. Hence,
\begin{equation*}
\sum _{i<j=1}^nd_{ij}\deltahatij =0
\end{equation*}
and evaluating at $\brac{\omega _r,\omega _s}$ gives $d_{rs}=0$, $r<s$.  This proves linear independence. To show that these $q$-measures span $\sscript (\ascript _n)$, let
$\mu\in\sscript (\ascript _n)$. Defining $I_{ij}^\mu$ as in \eqref{eq23}, we have that
\begin{equation}         
\label{eq24}
\mu =\sum _{i=1}^n\mu (\omega _i)\delta _i+\sum _{i<j=1}^nI_{ij}^\mu\deltahatij 
\end{equation}
because both sides agree on singleton and doubleton sets.
\end{proof}

We conclude from Theorem~\ref{thm22} that 
\begin{equation*}
\dim\sscript (\ascript _n)=n+\binom{n}{2}=n(n+1)/2
\end{equation*}
Also, if $\mu$ is a $q$-measure on $\ascript _n$, then $\mu$ has the unique representation given by \eqref{eq24}. The next theorem proves a result that we already mentioned for the particular case of a finite $\Omega$.

\begin{thm}       
\label{thm23}
If $\mu$ is a $q$-measure on $\ascript _n$, then there exists a unique symmetric signed measure
$\lambda$ on $\ascript _n\times\ascript _n$ such  $\mu (A)=\lambda (A\times A)$ for all
$A\in\ascript _n$.
\end{thm}
\begin{proof}
Define the function $\alpha\colon\Omega _n\times\Omega _n\to\real$ by
$\alpha (\omega _i,\omega _i)=\mu (\omega _i)$, $i=1,\ldots ,n$, and
\begin{equation*}
\alpha (\omega _i,\omega _j)=\alpha (\omega _j,\omega _i)=\tfrac{1}{2}I_{ij}^\mu
\end{equation*}
$i<j$, $i,j=1,\ldots ,n$. Define the signed measure $\lambda$ on $\ascript _n\times\ascript _n$ by
\begin{equation*}
\lambda (\Delta )=\brac{\sum\alpha (\omega _i,\omega _j)\colon (\omega _i,\omega _j)\in\Delta}
\end{equation*}
Now $\lambda$ is symmetric because
\begin{align*}
\lambda (A\times B)
  &=\brac{\sum\alpha (\omega _i,\omega _j)\colon (\omega _i,\omega _j)\in A\times B}\\
  &=\brac{\sum\alpha (\omega _i,\omega _j)\colon\omega _i\in A,\omega _j\in B}\\
  &=\brac{\sum\alpha (\omega _j,\omega _i)\colon\omega _i\in A,\omega _j\in B}\\
  &=\brac{\sum\alpha (\omega _j,\omega _i)\colon (\omega _j,\omega _i)\in B\times A}
  =\lambda (B\times A)\\
\end{align*}
We now show that $\lambda (A\times A)=\mu (A)$ for all $A\in\ascript _n$. We can assume without loss of generality that $A=\brac{\omega _1,\ldots ,\omega _m}$, $2\le m\le n$. It follows from
\eqref{eq22} that
\begin{equation*}
\mu (A)=\sum _{i<j=1}^m\mu\paren{\brac{\omega _i,\omega _j}}
  -(m-2)\sum _{i=1}^m\mu (\omega _i)
\end{equation*}
Moreover, we have that
\begin{align*}
\lambda (A\times A)&=\brac{\sum\alpha (\omega _i,\omega _j)\colon\omega _i,\omega _j\in A}\\
  &=\sum _{i=1}^m\alpha (\omega _i,\omega _i)+2\sum _{i<j=1}^m\alpha (\omega _i,\omega _j)\\
  &=\sum _{i=1}^m\mu (\omega _i)+\sum _{i<j=1}^mI_{ij}^\mu\\
  &=\sum _{i=1}^m\mu (\omega _i)+\sum _{i<j=1}^m\mu\paren{\brac{\omega _i,\omega _j}}
  -\sum _{i<j=1}^m\sqbrac{\mu (\omega _i)-\mu (\omega _j)}\\
  &=\sum _{i=1}^m\mu (\omega _i)+\sum _{i<j=1}^m\mu\paren{\brac{\omega _i,\omega _j}}
  -(m-1)\sum _{i=1}^m\mu (\omega _i)\\
  &=\sum _{i<j=1}^m\mu\paren{\brac{\omega _i,\omega _j}}-(m-2)\sum _{i=1}^m\mu (\omega _i)
  =\mu (A)
\end{align*}
To prove uniqueness, suppose $\lambda '\colon\ascript _n\times\ascript _n\to\real$ is a symmetric signed measure satisfying $\lambda '(A\times A)=\mu (A)$ for every $A\in\ascript _n$. We then have that
\begin{equation*}
\lambda '\sqbrac{(\omega _i,\omega _i)}=\lambda '\sqbrac{\brac{\omega _i}\times\brac{\omega _i}}
  =\mu (\omega _i)=\lambda\sqbrac{(\omega _i,\omega _i)}
\end{equation*}
for $i=1,\ldots ,n$. Moreover, letting $A=\brac{(\omega _i,\omega _j)}$ we have that
\begin{align*}
\mu (A)&=\lambda '(A\times A)
  =\lambda '\sqbrac{(\omega _i,\omega _i)}+\lambda '\sqbrac{(\omega _j,\omega _j)}
  +2\lambda '\sqbrac{(\omega _i,\omega _j)}\\
  &=\lambda\sqbrac{(\omega _i,\omega _i)}+\lambda\sqbrac{(\omega _j,\omega _j)}
  +2\lambda\sqbrac{(\omega _i,\omega _j)}
\end{align*}
Hence, $\lambda '\sqbrac{(\omega _i,\omega _j)}=\lambda\sqbrac{(\omega _i,\omega _j)}$ and since signed measures are determined by their values on singleton sets, we have that
$\lambda '=\lambda$.
\end{proof}

Let $(\Omega ,\ascript ,\mu )$ be a $q$-measure space and let $f\colon\Omega\to\real$ be a measurable function. The $q$-\textit{integral} of $f$ is defined by
\begin{equation}         
\label{eq25}
\int fd\mu =\int _0^\infty\mu\sqbrac{f^{-1}(\lambda ,\infty )}d\lambda
  -\int _0^\infty\mu\sqbrac{f^{-1}(-\infty ,-\lambda )}d\lambda
\end{equation}
where $d\lambda$ denotes Lebesgue measure on $\real$ \cite{gud091}. Any measurable function $f\colon\Omega\to\real$ has a unique representation $f=f_1-f_2$ where $f_1,f_2\ge 0$ are measurable and $f_1f_2=0$. It follows that
\begin{equation}         
\label{eq26}
\int fd\mu =\int f_1d\mu -\int f_2d\mu
\end{equation}
Because of \eqref{eq26} we only need to consider $q$-integrals of nonnegative functions. As usual in integration theory, if $A\in\ascript$ we define $\displaystyle{\int _Afd\mu =\int f\chi _Ad\mu}$ where $\chi _A$ is the characteristic function for $A$.

Any nonnegative simple measurable function $f\colon\Omega\to\real ^+$ has the canonical representation
\begin{equation}         
\label{eq27}
f=\sum _{i=1}^n\alpha _i\chi _{A_i}
\end{equation}
where $A_i\cap A_j=\emptyset$, $i\ne j$, $\cup A_i=\Omega$ and
$0\le\alpha _1<\alpha _2\cdots <\alpha _n$. It follows from \eqref{eq25} that
\begin{align}         
\label{eq28}
\int fd\mu =\alpha _1&\sqbrac{\mu\paren{\bigcupdotionen A_i}-\mu\paren{\bigcupdotitwon A_i}}
  \notag\\
 &+\cdots +\alpha _{n-1}\sqbrac{\mu\paren{A_{n-1}\cupdot A_n}-\mu (A_n)}
  +\alpha _n\mu (A_n)
\end{align}

\begin{exam}{1}    
Let $f\colon\Omega _n\to\real ^+$ be a nonnegative function on $\Omega _n$ and let
$\mu\colon\ascript _n\to\real ^+$ be a $q$-measure. Also, let
$\lambda\colon\ascript _n\times\ascript _n\to\real$ be the unique symmetric signed measure such that $\mu (A)=\lambda (A\times A)$ for every $A\in\ascript _n$. We can assume without loss of generality that
\begin{equation*}
0\le f(\omega _1)\le f(\omega _2)\le\cdots\le f(\omega _n)
\end{equation*}
By \eqref{eq28} we have for $i<j$ that
\begin{align*}
\int fd\deltahatij &=f(\omega _1)
  \sqbrac{\deltahatij\paren{\brac{\omega _1,\ldots ,\omega _n}}
  -\deltahatij\paren{\brac{\omega _2,\ldots ,\omega _n}}}\\
  &\quad +\cdots +f(\omega _{n-1})\sqbrac{\deltahatij\paren{\brac{\omega _{n-1},\omega _n}}
  -\deltahatij (\omega _n)}+f(\omega _n)\deltahatij (\omega _n)\\
  &=f(\omega _i)=\min\paren{f(\omega _i),f(\omega _j)}
\end{align*}
By \eqref{eq24} and the proof of Theorem~\ref{thm23}, it follows that
\begin{align}         
\label{eq29}
\int fd\mu
  &=\int fd\sqbrac{\sum _{i=1}^n\mu (\omega _i)\delta _i+\sum _{i<j=1}^nI_{ij}^\mu\deltahatij}
  \notag\\
  &=\sum _{i=1}^n\mu (\omega _i)f(\omega _i)
  +\sum _{i<j=1}^nI_{ij}^\mu\min\paren{f(\omega _i),f(\omega _j)}\notag\\
  &=\sum _{i,j=1}^n\min\paren{f(\omega _i),f(\omega _j)}\lambda\sqbrac{(\omega _i,\omega _j)}
  \notag\\
 &=\int\min\paren{f(\omega ),f(\omega ')}d\lambda (\omega ,\omega ')
\end{align}
\end{exam}

\begin{exam}{2}    
Let $\nu\colon\ascript _n\to\complex$ be a complex measure and define the $q$-measure
$\mu\colon\ascript _n\to\real ^+$ by $\mu (A)=\ab{\nu (A)}^2$. Then for $i<j$ we have
\begin{equation*}
I_{ij}^\mu=\ab{\nu (\omega _i)+\nu (\omega _j)}^2-\ab{\nu (\omega _i)}^2-\ab{\nu (\omega _j)}^2
  =2\rmre\nu (\omega _i)\overline{\nu (\omega _j)}
\end{equation*}
By \eqref{eq24} we conclude that
\begin{align*}
\mu&=\sum _{i=1}^n\mu (\omega _i)\delta _i
  +\sum _{i<j=1}^m\sqbrac{2\rmre\nu (\omega _i)\overline{\nu (\omega _j)}}\deltahatij\\
  &=\sum _{i=1}^n\ab{\nu (\omega _i)}^2\delta _i
  +2\rmre\sum _{i<j=1}^n\nu (\omega _i)\overline{\nu (\omega _j)}\delta _i\delta _j\\
  &=\ab{\sum _{i=1}^n\nu (\omega _i)\delta _i}^2
\end{align*}
As in Example~1, we obtain
\begin{equation*}
\int fd\mu =\sum _{i=1}^n\mu (\omega _i)f(\omega _i)
  +\sum _{i<j=1}^n\sqbrac{2\rmre\nu (\omega _i)\overline{\nu (\omega _j)}}
  \min\paren{f(\omega _i),f(\omega _j)}
\end{equation*}
\end{exam}
\medskip

Equation~\eqref{eq29} suggests the following new characterization of the $q$-integral. If
$\int\ab{f}d\mu <\infty$, then $f$ is \textit{integrable}.

\begin{thm}       
\label{thm24}
Let $(\Omega ,\ascript ,\mu )$ be a $q$-measure space and let $f\colon\Omega\to\real ^+$ be integrable. If $\lambda\colon\ascript\times\ascript\to\real$ is the unique symmetric signed measure such that $\mu (A)=\lambda (A\times A)$ for all $A\in\ascript$, then
\begin{equation}         
\label{eq210}
\int fd\mu =\int\min\paren{f(\omega ),f(\omega ')}d\lambda (\omega ,\omega ')
\end{equation}
\end{thm}
\begin{proof}
Let $f$ be a nonnegative simple function on $\Omega$ with canonical representation \eqref{eq27}. If $g(\omega ,\omega ')=\min\paren{f(\omega ),f(\omega ')}$, then when $\omega\in A_i$,
$\omega '\in A_j$ we have that
\begin{equation*}
g(\omega ,\omega ')=
\begin{cases}
\alpha _i&\hbox{if }i\le j\\\alpha _j&\hbox{if }j<i
\end{cases}
\end{equation*}
Hence,
\begin{equation*}
g=\sum _{i\le j=1}^n\alpha _i\chi _{A_i\times A_j}+\sum _{i>j=1}^n\alpha _j\chi _{A_i\times A_j}
\end{equation*}
It follows that
\begin{align*}
\int g(\omega ,\omega ')d\lambda (\omega ,\omega ')
  &=\sum _{i\le j=1}^n\alpha _i\lambda (A_i\times A_i)
  +\sum _{i>j=1}^n\alpha _j\lambda (A_i\times A_j)\\
  &=\alpha _1
  \sqbrac{\lambda (A_1\times A_1)+2\lambda (A_1\times A_2)+\cdots +2\lambda (A_1\times A_n)}\\
  &\quad +\alpha _2
  \sqbrac{\lambda (A_2\times A_2)+2\lambda (A_2\times A_3)+\cdots +2\lambda (A_2\times A_n)}\\
 &\quad\ \vdots\\
 &\quad +\alpha _{n-1}
 \sqbrac{\lambda (A_{n-1}\times _{n-1})+2\lambda (A_{n-1}\times A_n)}\\
 &\quad +\alpha _n\lambda (A_n\times A_n)
\end{align*}
On the other hand, by \eqref{eq28} and grade-2 additivity we have
\begin{align*}
\int fd\mu&=\alpha _1\left [\mu (A_1\cupdot A_2)+\cdots +\mu (A_1\cupdot A_n)\right.\\
  &\quad\left.-(n-1)\mu (A_1)-\mu (A_2)-\cdots -\mu (A_n)\right ]\\
  &\quad +\alpha _2\left [\mu (A_2\cupdot A_3)+\cdots +\mu (A_2\cupdot A_n)\right.\\
  &\quad\left.-(n-2)\mu (A_2)-\mu (A_3)-\cdots -\mu (A_n)\right ]\\
  &\quad\ \vdots\\
  &\quad +\alpha _{n-1}\sqbrac{\mu (A_{n-1}\cupdot A_n)-\mu (A_n)}+\alpha _n\mu (A_n)\\
  &=\alpha _1\left [\lambda (A_1\cupdot A_2\times A_1\cupdot A_2)+\cdots +
  \lambda (A_1\cupdot A_n\times A_1\cupdot A_n)\right.\\
  &\quad\left.-(n-1)\mu (A_1)-\mu (A_2)-\cdots -\mu (A_n)\right ]\\
  &\quad +\alpha _2\left [\lambda (A_2\cupdot A_3\times A_2\cupdot A_3)+\cdots +
  \lambda (A_2\cupdot A_n\times A_2\cupdot A_n)\right.\\
  &\quad\left.-(n-2)\mu (A_2)-\mu (A_3)-\cdots -\mu (A_n)\right ]\\
  &\quad\ \vdots\\
  &\quad +\alpha _{n-1}\sqbrac{\lambda (A_{n-1}\cupdot A_n\times A_{n-1}\cupdot A_n)-\mu (A_n)}
  +\alpha _n\mu (A_n)\\
  &=\alpha _1\left [\lambda (A_1\times A_1)+2\lambda (A_1\times A_2)+\lambda (A_2\times A_2)
  +\cdots +\lambda (A_1\times A_1)\right.\\
  &\quad\left.+2\lambda (A_1\times A_n)+\lambda (A_n\times A_n)\right.\\
  &\quad\left.-(n-1)\mu (A_1)-\mu (A_2)-\cdots -\mu (A_n)\right ]\\
  &\quad +\alpha _2\left [\lambda (A_2\times A_2)+2(A_2\times A_3)+\lambda (A_3\times A_3)
  +\cdots +\lambda (A_2\times A_2)\right.\\
  &\quad\left. +2\lambda (A_2\times A_n)+\lambda (A_n\times A_n)\right.\\
  &\quad\left. -(n-2)\mu (A_2)-\mu (A_3)-\cdots -\mu (A_n)\right ]\\
  &\quad\ \vdots\\
  &\quad +\alpha _{n-1}\sqbrac{\lambda (A_{n-1}\times A_{n-1})+2\lambda (A_{n-1}\times A_n)
  +\lambda (A_n\times A_n)-\mu (A_n)}\\
  &\quad +\alpha _n\mu (A_n)
\end{align*}
which reduces to the expression given for $\int g(\omega ,\omega ')d\lambda (\omega ,\omega ')$. We conclude that \eqref{eq210} holds for nonnegative measurable simple functions. Since $f$ is the limit of an increasing sequence of such functions, the result follows from the quantum dominated monotone convergence theorem \cite{gud091}.
\end{proof}

We now apply Theorem~\ref{thm24} to compute some $q$-integrals for an interesting
$q$-measure. Let  $\Omega =[0,1]\subseteq\real$, let $\nu$ be Lebesgue measure on $\Omega$ and define the $q$-measure $\mu$ on the Borel $\sigma$-algebra $\bscript (\real )$ by
$\mu (A)=\nu (A)^2$. We call $\mu$ the \textit{quantum (Lebesgue)}${}^2$ \textit{measure}. By Theorem~\ref{thm24}, if $f\colon\Omega\to\real ^+$ is integrable and $0\le a<b\le 1$, we have
\begin{equation}         
\label{eq211}
\int _a^bfd\mu =\int _{[a,b]}fd\mu =\int _a^b\int _a^b\min\paren{f(x),f(y)}dxdy
\end{equation}
Suppose $f$ is increasing on $\Omega$ and let $g_y(x)=\min\paren{f(x),f(y)}$. Then
\begin{equation*}
g_u(x)=
\begin{cases}
f(x)&\hbox{for  }x\le y\\ f(y)&\hbox{for }x\ge y
\end{cases}
\end{equation*}
Hence, by \eqref{eq211} we have
\begin{equation}         
\label{eq212}
\int _a^bfd\mu =\int _a^b\int _a^bg_y(x)dxdy=\int _a^b\sqbrac{\int _a^yf(x)dx+f(y)(b-y)}dy
\end{equation}
Since
\begin{equation*}
\int _a^b\int _a^yf(x)dxdy=\int _a^b\int _x^bf(x)dydx=\int _a^bf(x)(b-x)dx
\end{equation*}
we conclude from \eqref{eq212} that
\begin{equation}         
\label{eq213}
\int _a^bfd\mu =2\int _a^b\int _a^yf(x)dxdy=2\int _a^bf(x)(b-x)dx
\end{equation}
Similarly, if $f$ is decreasing, then
\begin{equation}         
\label{eq214}
\int _a^bfd\mu =\int _a^b\sqbrac{f(y)(y-a)+\int _y^bf(x)dx}dy
\end{equation}
Since
\begin{equation*}
\int _a^b\int _y^bf(x)dxdy=\int _a^b\int _a^xf(x)dydx=\int _a^bf(x)(x-a)dx
\end{equation*}
\eqref{eq214} becomes
\begin{equation}         
\label{eq215}
\int _a^bfd\mu =2\int _a^b\int _y^bf(x)dxdy=2\int _a^bf(x)(x-a)dx
\end{equation}

\begin{exam}{3}    
For $n\ge 0$, since $f(x)=x^n$ is increasing, by \eqref{eq213} we have
\begin{align*}
\int _a^bx^nd\mu&=2\int _a^bx^n(b-x)dx\\
  &=\frac{2}{(n+1)(n+2)}\sqbrac{b^{n+2}-a^{n+1}\paren{(n+2)b-(n+1)a}}
\end{align*}
\end{exam}

\begin{exam}{4}    
Since $f(x)=e^x$ is increasing, by \eqref{eq213} we have
\begin{equation*}
\int _a^be^xdx=2\int _a^be^x(b-x)dx=2e^b-2e^a(b-a+1)
\end{equation*}
\end{exam}

\begin{exam}{5}    
For $n\ge 3$, since $f(x)=x^{-n}$ is decreasing, by (2.15) we have
\begin{align*}
\int _a^bx^{-n}d\mu&=2\int _a^bx^{-n}(x-a)dx\\
  &=\frac{2}{(n-1)(n-2)}\sqbrac{b^{-n+1}\paren{(n-2)a-(n-1)b}+a^{-n+2}}
\end{align*}
\end{exam}

\section{Anhomomorphic Logics} 
In this and the next section we shall restrict our attention to a finite outcome space
$\Omega _n=\brac{\omega _1,\ldots ,\omega _n}$ and its corresponding set of events
$\ascript _n=P(\Omega _n)$. Let $\integers _2$ be the two element Boolean algebra $\brac{0,1}$ with the usual multiplication and with addition given by $0\oplus 1=1\oplus 0=1$ and
$0\oplus 0=1\oplus 1=0$. A \textit{coevent} on $\ascript _n$ is a truth function
$\phi\colon\ascript _n\to\integers _2$ such that $\phi (\emptyset )=0$ \cite{gt09, sor941, sor942}. A coevent $\phi$ corresponds to a potential reality for a physical system in the sense that $\phi (A)=1$ if $A$ occurs and $\phi (A)=0$ if $A$ does not occur. For a classical system a coevent
$\phi$ is taken to be a \textit{homomorphism} by which we mean that

\begin{list} {(H\arabic{cond})}{\usecounter{cond}
\setlength{\rightmargin}{\leftmargin}}
\item $\phi (\Omega _n)=1$\enspace (unital)
\item $\phi (A\cupdot B)=\phi (A)\oplus\phi (B)$\enspace (additive)
\item $\phi (A\cap B)=\phi (A)\phi (B)$\enspace (multiplicative)
\end{list}

Since there are various quantum systems for which the truth function is not a homomorphism
\cite{gt09, sor941}, at least one of these condition must fail. Denoting the set of coevents by
$\ascript _n^*$, since the elements of $\ascript _n^*$ are not all homomorphisms we call
$\ascript _n^*$ the \textit{full anhomomorphic logic} \cite{gud092}. Notice that the cardinality
$\ab{\ascript _n^*}=2^{2^n-1}$ is very large.

Corresponding to $\omega _i\in\Omega _n$ we define the \textit{evaluation map} $\omega _i^*\colon\ascript _n\to\integers _2$ by $\omega _i^*(A)=1$ if and only if $\omega _i\in A$. It can be shown that a coevent $\phi$ is a homomorphism if and only if $\phi =\omega _i^*$ for some $i=1,\ldots ,n$ \cite{gt09, gud092}. Thus, there are only $n$ possible realities for a classical system. For two coevents $\phi$, $\psi$ we define their sum and product in the usual way by
$(\phi\oplus\psi )(A)=\phi (A)\oplus\psi (A)$ and $(\phi\psi )(A)=\phi (A)\psi (A)$ for all
$A\in\ascript _n$. It can be shown that any coevent has a unique representation (up to order of the terms) as a polynomial in the evaluation maps.

A coevent that satisfies (H2) is called \textit{additive} and $\phi$ is additive if and only if it has the form
\begin{equation*}
\phi = c_1\omega _1^*\oplus\cdots\oplus c_n\omega _n^*
\end{equation*}
where $c_i=0$ or $1$, $i=1,\ldots ,n$ \cite{gt09, gud092, sor941}. Denoting the set of additive coevents by $\ascript _{n,a}^*$, we see that $\ab{\ascript _{n,a}^*}=2^n$. A coevent that satisfies (H3) is called \textit{multiplicative} and $\phi$ is multiplicative if and only if it has the form
\begin{equation*}
\phi =\omega _{i_1}^*\omega _{i_2}^*\cdots\omega _{i_m}^*
\end{equation*}
where, by convention, $\phi =0$ if there are no terms in the product \cite{gt09, gud092, sor09}. Denoting the set of multiplicative coevents by $\ascript _{n,m}^*$ we again have that
$\ab{\ascript _{n,m}^*}=2^n$. A coevent $\phi$ is \textit{quadratic} or \textit{grade}-2 if it satisfies
\begin{equation*}
\phi (A\cupdot B\cupdot C)=\phi (A\cupdot B)\oplus\phi (A\cupdot C)\oplus\phi (B\cupdot C)
  \oplus\phi (A)\oplus\phi (B)\oplus\phi (C)
\end{equation*}
It can be shown that $\phi$ is quadratic if and only  if it has the form
\begin{equation*}
\phi = c_1\omega _1^*\oplus\cdots\oplus c_n\omega _n^*\oplus d_{12}\omega _1^*\omega _2^*
  \oplus\cdots\oplus d_{n-1,n}\omega _{n-1}^*\omega _n^*
\end{equation*}
where $c_i=0$ or $1$, $i=1,\ldots ,n$, and $d_{ij}=0$ or $1$, $i<j$, $i,j=1,\ldots ,n$
\cite{gt09, gud092}. Denoting the set of quadratic coevents by $\ascript _{n,q}^*$, we have that
$\ab{\ascript _{n,q}^*}=2^{n(n+1)/2}$.

Let $\bscript$ be a Boolean algebra of subsets of a set. The set-theoretic operations on $\bscript$ are
$\cup$, $\cap$ and ${}'$ where $A'$ denotes the complement of $A\in\bscript$. A subset
$\bscript _0\subseteq\bscript$ is called a \textit{subalgebra} (or \textit{subring}) of $\bscript$ if
$\emptyset\in\bscript _0$, $A\cup B$, $A\cap B$, $A\smallsetminus B\in\bscript _0$ whenever
$A,B\in\bscript _0$ where $A\smallsetminus B=A\cap B'$. If $\bscript _0$ is finite, then $\bscript _0$ has a largest element $C$ and $\bscript _0$ is itself a Boolean algebra under the operations $\cup$,
$\cap$ and complement $A'=C\smallsetminus A$. 

For $\phi\in\ascript _n^*$, it is of interest to find subalgebras of $\ascript _n$ on which $\phi$ acts classically. If $\bscript _0$ is a subalgebra of $\ascript _n$ and the restriction $\phi\mid\bscript _0$ is a homomorphism, then $\bscript _0$ is a \textit{classical subdomain} for $\phi$. A
\textit{classical domain} $\bscript$ for $\phi$ is a maximal classical subdomain for $\phi$. That is,
$\bscript$ is a classical subdomain for $\phi$ and if $\cscript$ is a classical subdomain for $\phi$ with
$\bscript\subseteq\cscript$, then $\bscript =\cscript$.  Any classical subdomain is contained in a classical domain $\bscript$ but $\bscript$ need not be unique. The rest of this section is devoted to the study of classical domains. For $\phi\in\ascript _n^*$, the $\phi$-\textit{center} $Z_\phi$ is the set of elements $A\in\ascript _n$ such that
\begin{equation}         
\label{eq31}
\phi (B)=\phi (B\cap A)\oplus\phi (B\cap A')
\end{equation}
for all $B\in\ascript _n$.

\begin{thm}       
\label{thm31}
$Z_\phi$ is a subalgebra of $\ascript _n$ and $\phi\mid Z_\phi$ is additive.
\end{thm}
\begin{proof}
It is clear that $\emptyset ,\Omega _n\in Z_\phi$ and that $A'\in Z_\phi$ whenever $A\in Z_\phi$. Suppose that $A,B\in Z_\phi$. We shall show that $A\cap B\in Z_\phi$. Since $B\in Z_\phi$, we have that
\begin{equation*}
\phi (C\cap A)=\phi (C\cap A\cap B)\oplus\phi (C\cap A\cap B')
\end{equation*}
Hence,
\begin{equation*}
\phi (C\cap A\cap B)=\phi (C\cap A)\oplus\phi (C\cap A\cap B')
\end{equation*}
for all $C\in\ascript _n$. It follows that
\begin{align*}
\phi (C\cap A\cap&B)\oplus\phi \paren{C\cap (A\cap B)'}
  =\phi (C\cap A\cap B)\oplus\phi\paren{C\cap (A'\cup B')}\\
  &=\phi (C\cap A)\oplus\phi (C\cap A\cap B')\oplus\phi\sqbrac{(C\cap A')\cup (C\cap B')}\\
  &=\phi (C)\oplus\phi (C\cap A')\oplus\phi (C\cap A\cap B')\oplus\phi
  \sqbrac{(C\cap A')\cup (C\cap B')}\\
  &=\phi (C)\oplus\phi (C\cap A')\oplus\phi (C\cap A\cap B')\oplus\phi (C\cap A\cap B')\\
  &\quad\oplus\phi\sqbrac{(C\cap A')\cup (C\cap A'\cap B')}\\
  &=\phi (C)\oplus\phi (C\cap A')\oplus\phi (C\cap A')=\phi (C)
\end{align*}
Hence, $A\cap B\in Z_\phi$. Moreover, if $A,B\in Z_\phi$, then $A',B'\in Z_\phi$ so
$A'\cap B'\in Z_\phi$. Hence, $A\cup B=(A'\cap B')'\in Z_\phi$. It follows that $Z_\phi$ is a subalgebra of $\ascript _n$. To show that $\phi\mid Z_\phi$ is additive, suppose that
$A,B\in Z_\phi$ with $A\cap B=\emptyset$. We conclude that
\begin{equation*}
\phi (A\cupdot B)=\phi\sqbrac{(A\cupdot B)\cap A}\oplus\phi\sqbrac{(A\cupdot B)\cap A'}
=\phi (A)\oplus\phi (B)\qedhere
\end{equation*}
\end{proof}

A general $\phi\in\ascript _n^*$ may have many classical domains and these seem to be tedious to find. However, $Z_\phi$ is relatively easy to find and we now give a method for constructing classical subdomains within $Z_\phi$. Since $\phi\mid Z_\phi$ is additive, we are partly there and only need to find a subalgebra $Z_\phi ^1$ of $Z_\phi$ on which $\phi$ is multiplicative and not $0$. Indeed, the unital condition (H1) then holds because there exists an $A$ such that $\phi (A)\ne 0$ and hence,
\begin{equation*}
\phi (A)=\phi (A\cap C)=\phi (A)\phi (C)
\end{equation*}
which implies that $\phi (C)=1$ where $C$ is the largest element of $Z_\phi ^1$.

An \textit{atom} in a Boolean algebra is a minimal nonzero element. Let $A_1,\ldots , A_m$ be the atoms of $Z_\phi$. Then $A_1,\ldots ,A_m$ are mutually disjoint, nonempty and
$\cupdot A_i=\Omega _n$. Moreover, every nonempty set in $Z_\phi$ has the form
$B=\cupdot _{j=1}^rA_{i_j}$. Define $A_i^*\in\ascript _n^*$ by $A_i^*(A)=1$ if and only if
$A_i\subseteq A$, $i=1,\ldots ,m$. Since $\phi$ is additive on $Z_\phi$, it follows that $\phi$ has the form
\begin{equation*}
\phi =A_{i_1}^*\oplus\cdots\oplus A_{i_r}^*
\end{equation*}
on $Z_\phi$. We can assume without loss of generality that
\begin{equation*}
\phi =A_1^*\oplus\cdots\oplus A_r^*
\end{equation*}
Let $Z_\phi ^i$ be the subalgebra of $Z_\phi$ generated by $A_i, A_{r+1},\ldots ,A_m$,
$i=1,\ldots ,r$. Then $\phi\mid Z_\phi ^i=A_i^*$, $i=1,\ldots ,r$.

\begin{cor}       
\label{cor32}
If $\phi\ne 0$, then $Z_\phi ^i$ is a classical subdomain for $\phi$, $i=1,\ldots ,r$.
\end{cor}
\begin{proof}
For simplicity, we work with $Z_\phi ^1$ and the other $Z_\phi ^i$, $i=2,\ldots ,r$ are similar. Now
$Z_\phi ^1$ is a subalgebra of $\ascript _n$ with largest element
\begin{equation*}
B_1=A_1\cupdot A_{r+1}\cupdot A_{r+2}\cupdot\cdots\cupdot A_m
\end{equation*}
Now $\phi\mid Z_\phi ^1=A_1^*$, $A_1^*$ is additive and $A_1^*(B_1)=1$. To show that $A_1^*$ is multiplicative, we have for $A,B\in Z_\phi ^1$ that $A_1^*(A\cap B)=1$ if and only if
$A_1\subseteq A\cap B$. Since $A_1^*(A)A_1^*(B)=1$ if and only if $A_1\subseteq A$ and $A_1\subseteq B$, we have that $A_1^*(A)A_1^*(B)=1$ if and only if $A_1\subseteq A\cap B$. Hence, $A_1^*(A\cap B)=A_1^*(A)A_1^*(B)$. We conclude that $\phi\mid Z_\phi ^1=A_1^*$ is a homomorphism on $Z_\phi ^1$.
\end{proof}

\begin{exam}{6}    
We consider some coevents in $\ascript _3^*$. For $\phi =\omega _1^*\omega _2^*$,
\begin{equation*}
Z_\phi =\brac{\phi,\brac{\omega _3},\brac{\omega _1,\omega _2},\Omega _3}
\end{equation*}
and $Z_\phi$ is the unique classical domain for $\phi$. For
$\psi =\omega _1^*\oplus\omega _2^*$, $Z_\psi =\ascript _3$ and the two classical domains for
$\psi$ are
\begin{equation*}
\brac{\phi,\brac{\omega _1},\brac{\omega _3},\brac{\omega _1,\omega _3}},\quad
\brac{\phi ,\brac{\omega _2},\brac{\omega _3},\brac{\omega _2,\omega _3}}
\end{equation*}
For $\gamma =\omega _1^*\oplus\omega _2^*\oplus\omega _3^*$, $Z_\gamma =\ascript _3$,
$Z_\gamma ^1=\brac{\phi ,\brac{\omega _1}}$, $Z_\gamma ^2=\brac{\phi ,\brac{\omega _2}}$,
$Z_\gamma ^3=\brac{\phi ,\brac{\omega _3}}$. The classical domains for $\gamma$ are
\begin{equation*}
\brac{\phi,\brac{\omega _1},\brac{\omega _2,\omega _3},\Omega _3},\quad
\brac{\phi ,\brac{\omega _2},\brac{\omega _1,\omega _3},\Omega _3},\quad
\brac{\phi ,\brac{\omega _3},\brac{\omega _1,\omega _2},\Omega _3}
\end{equation*}
For $\delta =\omega _1^*\omega _2^*\omega _3^*$, the unique classical domain is
$Z_\delta =\brac{\phi ,\Omega _3}$.
\end{exam}

\section{Transferring Quantum Measures} 
This section provides a connection between the previous two sections. We study a method for transferring a $q$-measure on $\ascript _n$ to a measure on $\ascript _n^*$. A $q$-measure
$\mu$ on $\ascript _n$ is \textit{pure} if the values of $\mu$ are contained in $\brac{0,1}$. Let
$\mscript (\ascript _n)$ be the set of $q$-measures $\mu$ on $\ascript _n$ such that
\begin{equation*}
\max\brac{\mu (A)\colon A\in\ascript _n}\le 1
\end{equation*}
We have seen in Section~2 that the set of signed $q$-measures on $\ascript _n$ forms a finite dimensional real linear space $\sscript (\ascript _n)$ and it is clear that $\mscript (\ascript _n)$ is a convex subset of $\sscript (\ascript _n)$. Letting $\pscript (\ascript _n)$ be the set of pure
$q$-measures, we now show that the elements of $\pscript (\ascript _n)$ are extremal in
$\mscript (\ascript _n)$. Let $\mu\in\pscript (\ascript _n)$ and suppose that
$\mu =\lambda\mu _1+(1-\lambda )\mu _2$ for $0<\lambda <1$ and
$\mu _1,\mu _2\in\mscript (\ascript _n)$. If $\mu (A)=0$, then
\begin{equation*}
\lambda\mu _1(A)+(1-\lambda )\mu _2(A)=0
\end{equation*}
so that $\mu _1(A)=\mu _2(A)=0$. If $\mu (A)=1$, then
\begin{equation*}
\lambda\mu _1(A)+(1-\lambda )\mu (A)=1
\end{equation*}
so that $\mu _1(A)=\mu _2(A)=1$. Hence, $\mu _1=\mu _2=\mu$ which shows that $\mu$ is extremal.

We next show that the extremal elements of $\mscript (\ascript _n)$ are in $\pscript (\ascript _n)$. To illustrate this, let $\mu\in\mscript (\ascript _2)$ be extremal. If $\mu\notin\pscript (\ascript _2)$, then at least one of the numbers $\mu (\omega _1)$, $\mu (\omega _2)$, $\mu (\Omega _2)$ is not $0$ or $1$. Suppose, for example that $0<\mu (\omega _1)<1$. Then there exists an
$\varepsilon >0$ such that $\varepsilon <\mu (\omega _1)<1-\varepsilon$. Define
$\mu _1\colon\ascript _2\to\real ^+$ by $\mu _1(A)=\mu (A)$ if $A\ne\brac{\omega _1}$ and
$\mu _1(\omega _1)=\mu (\omega _1)+\varepsilon$. Also, define
$\mu _2\colon\ascript _2\to\real ^+$ by $\mu _2(A)=\mu (A)$ if $A\ne\brac{\omega _1}$ and
$\mu _2(\omega _1)=\mu (\omega _1)-\varepsilon$. Then $\mu _1,\mu _2\in\mscript (\ascript _2)$,
$\mu _1\ne\mu _2$ and $\mu =\frac{1}{2}\mu _1+\frac{1}{2}\mu _2$. This contradicts the fact that $\mu$ is extremal. Hence, $\mu\in\pscript (\ascript _2)$. 

\begin{thm}       
\label{thm41}
$\pscript (\ascript _n)$ is precisely the set of extremal elements of $\mscript (\ascript _n)$.
\end{thm}
\begin{proof}
Denoting the set of extremal elements of $\mscript (\ascript _n)$ by $\rmext\mscript (\ascript _n)$ we have already shown that $\pscript (\ascript _n)\subseteq\rmext\mscript (\ascript _n)$. The proof that $\rmext\mscript (\ascript _n)\!\subseteq\pscript (\ascript _n)$ in general appears to be quite tedious so we shall indicate the proof for $\ascript _3$ and $\ascript _4$ whereby the reader can see the pattern. To show that $\rmext\mscript (\ascript _3)\subseteq\pscript (\ascript _3)$ suppose that
$\mu\in\rmext\mscript (\ascript _3)$ and $\mu\notin\pscript (\ascript _3)$. Then one of the numbers
$\mu (A)$, $A\in\ascript\ _3\smallsetminus\brac{\emptyset}$ is not $0$ or $1$. Since
\begin{equation}         
\label{eq41}
\mu (\Omega _3)=\mu\paren{\brac{\omega _1,\omega _2}}
  +\mu\paren{\brac{\omega _1,\omega _3}}+\mu\paren{\brac{\omega _2,\omega _3}}
  -\mu (\omega _1)-\mu (\omega _2)-\mu (\omega _3)
\end{equation}
a second of these numbers is not $0$ or $1$. There are various possibilities, all of which are similar, and we shall consider two of them. Suppose, for example, that $0<\mu\paren{\brac{\omega _1,\omega _2}}<1$ and $0<\mu (\omega _3)<1$. Then there exists an $\varepsilon >0$ such that
$\varepsilon <\mu\paren{\brac{\omega _2,\omega _2}}<1-\varepsilon$ and
$\varepsilon <\mu (\omega _3)<1-\varepsilon$. Define $\mu _1\colon\ascript _3\to\real ^+$ by
$\mu _1(A)=\mu (A)$ for $A\ne\brac{\omega _1,\omega _2}$ or $\brac{\omega _3}$ and
$\mu _1(\omega _3)=\mu (\omega _3)+\varepsilon$,
$\mu _1\paren{\brac{\omega _1,\omega _2}}
=\mu\paren{\brac{\omega _1,\omega _2}}+\varepsilon$.
Thus, $\mu _1$ satisfies \eqref{eq41} so $\mu _1\in\mscript (\ascript _3)$. Define
$\mu _2\colon\ascript _3\to\real ^+$ by $\mu _2(A)=\mu (A)$ for
$A\ne\brac{\omega _1,\omega _2}$ or $\brac{\omega _3}$ and
$\mu _2(\omega _3)=\mu (\omega _3)-\varepsilon$,
$\mu _2\paren{\brac{\omega _1,\omega _2}}=\mu\paren{\brac{\omega _1,\omega _2}}-\varepsilon$.
Again $\mu _2$ satisfies \eqref{eq41} so $\mu _2\in\mscript (\ascript _3)$. Also, $\mu _1\ne\mu _2$ and $\mu =\frac{1}{2}\mu _1+\frac{1}{2}\mu _2$ which contradicts the fact that
$\mu\in\rmext\mscript (\ascript _3)$. As another case, suppose that $0<\mu (\omega _1)<1$ and $0<\mu (\Omega _3)<1$. Then there exists an $\varepsilon <\mu (\omega )<1-\varepsilon$ and
$\varepsilon <\mu (\Omega _3)<1-\varepsilon$. Define $\mu _1,\mu _2\colon\ascript _3\to\real ^+$ by $\mu _1(A)=\mu _2(A)=\mu (A)$ for all $A\ne\brac{\omega _1}$ or $\Omega _3$ and
$\mu _1(\omega _1)=\mu (\omega _1)+\varepsilon$,
$\mu _1(\Omega _3)=\mu (\Omega _3)-\varepsilon$,
$\mu _2(\omega _1)=\mu (\omega _1)-\varepsilon$, $\mu _2(\Omega _3)+\varepsilon$. Then
$\mu _1,\mu _2$ satisfy \eqref{eq41} so $\mu _1,\mu _2\in\mscript (\ascript _3)$. Moreover,
$\mu _1\ne\mu _2$ and $\mu =\frac{1}{2}\mu _1+\frac{1}{2}\mu _2$ which contradicts the fact that
$\mu\in\rmext\mscript (\ascript _3)$. Since this method applies to all the cases, we conclude that
$\rmext\mscript (\ascript _3)\subseteq\pscript (\ascript _3)$.

To show that $\rmext\mscript (\ascript _4)\subseteq\pscript (\ascript _4)$ suppose that
$\mu\in\rmext\mscript (\ascript _4)$ and $\mu\notin\pscript (\mscript _4)$. By grade-2 additivity we have that
\begin{equation}         
\label{eq42}
\mu (\Omega _4)=\sum _{i<j=1}^4\mu\paren{\brac{\omega _i,\omega _j}}
  -2\sum _{i=1}^4\mu (\omega _i)
\end{equation}
\begin{align}         
\label{eq43}
\mu&\paren{\brac{\omega _1,\omega _2,\omega _3}}\\
&=\mu\paren{\brac{\omega _1,\omega _2}}
  +\mu\paren{\brac{\omega _1,\omega _3}}+\mu\paren{\brac{\omega _2,\omega _3}}
  -\mu (\omega _1)-\mu (\omega _2)-\mu (\omega _3)\notag\\
\label{eq44}            
\mu&\paren{\brac{\omega _1,\omega _2,\omega _4}}\\
&=\mu\paren{\brac{\omega _1,\omega _2}}
  +\mu\paren{\brac{\omega _1,\omega _4}}+\mu\paren{\brac{\omega _2,\omega _4}}
  -\mu (\omega _1)-\mu (\omega _2)-\mu (\omega _4)\notag\\
\label{eq45}            
\mu&\paren{\brac{\omega _1,\omega _3,\omega _4}}\\
&=\mu\paren{\brac{\omega _1,\omega _3}}
  +\mu\paren{\brac{\omega _1,\omega _4}}+\mu\paren{\brac{\omega _3,\omega _4}}
  -\mu (\omega _1)-\mu (\omega _3)-\mu (\omega _4)\notag\\
\label{eq46}            
\mu&\paren{\brac{\omega _2,\omega _3,\omega _4}}\\
&=\mu\paren{\brac{\omega _2,\omega _3}}
  +\mu\paren{\brac{\omega _2,\omega _4}}+\mu\paren{\brac{\omega _3,\omega _4}}
  -\mu (\omega _2)-\mu (\omega _3)-\mu (\omega _4)\notag
\end{align}
Since $\mu\notin\pscript (\ascript _4)$ we have that $\varepsilon >0$ where 
\begin{equation*}
\varepsilon =\min\brac{\mu (A),1-\mu (A)\colon A\in\ascript _4,\mu (A)\ne 0,1}
\end{equation*}
There are many possibilities and we shall again treat two of them. In \eqref{eq43}--\eqref{eq46} if there is one term that is not $0$ or $1$, then there are at least two such terms. Suppose that
$\mu\paren{\brac{\omega _1,\omega _2}}$, $\mu\paren{\brac{\omega _1,\omega _3}}$,
$\mu\paren{\brac{\omega _2,\omega _4}}$, $\mu\paren{\brac{\omega _3,\omega _4}}$ are not $0$ or $1$. Define $\mu _1\colon\ascript _4\to\real ^+$ by $\mu _1(A)=\mu (A)$ if $A$ is not one of these four given sets and
\begin{align*}
\mu _1\paren{\brac{\omega _1,\omega _2}}&=\mu\paren{\brac{\omega _1,\omega _2}}+\varepsilon\\
\mu _1\paren{\brac{\omega _1,\omega _3}}&=\mu\paren{\brac{\omega _1,\omega _3}}-\varepsilon\\
\mu _1\paren{\brac{\omega _2,\omega _4}}&=\mu\paren{\brac{\omega _2,\omega _4}}-\varepsilon\\
\mu _1\paren{\brac{\omega _3,\omega _4}}&=\mu\paren{\brac{\omega _3,\omega _4}}+\varepsilon
\end{align*}
Then $\mu _1$ satisfies \eqref{eq42}--\eqref{eq46} so $\mu _1\in\mscript (\ascript _4)$. Define
$\mu _2\colon\ascript _4\to\real ^+$ the same as $\mu _1$ except that $\varepsilon$ is replaced by
$-\varepsilon$. Again, $\mu _2$ satisfies \eqref{eq42}--\eqref{eq46} so
$\mu _2\in\mscript (\ascript _4)$. Moreover, $\mu _1\ne\mu _2$ and
$\mu =\frac{1}{2}\mu _1+\frac{1}{2}\mu _2$ which contradicts the fact that
$\mu\in\rmext\mscript (\ascript _4)$. For our final case, suppose that $\mu (\omega _3)$,
$\mu\paren{\brac{\omega _1,\omega _2}}$, $\mu\paren{\brac{\omega _3,\omega _4}}$,
$\mu\paren{\brac{\omega _1,\omega _2,\omega _4}}$ are not $0$ or $1$. Define
$\mu _1\colon\ascript _4\to\real ^+$ by $\mu _1(A)=\mu (A)$ if $A$ is not one of these four given sets and $\mu _1(\omega _3)=\mu (\omega _3)+\varepsilon$
\begin{align*}
\mu _1\paren{\brac{\omega _1,\omega _2}}&=\mu\paren{\brac{\omega _1,\omega _2}}+\varepsilon\\
\mu _1\paren{\brac{\omega _3,\omega _4}}&=\mu\paren{\brac{\omega _3,\omega _4}}+\varepsilon\\
\mu _1\paren{\brac{\omega _1,\omega _2,\omega _4}}
  &=\mu\paren{\brac{\omega _1,\omega _2,\omega _4}}+\varepsilon
\end{align*}
As before $\mu _1$ satisfies \eqref{eq42}--\eqref{eq46} so $\mu _1\in\mscript (\ascript _4)$. Define
$\mu _2\colon\ascript _4\to\real ^+$ the same as $\mu _1$ except $\varepsilon$ is replaced by
$-\varepsilon$. Again, $\mu _2$ satisfies \eqref{eq42}--\eqref{eq46} so
$\mu _2\in\mscript (\ascript _4)$. Moreover, $\mu _1\ne\mu _2$ and
$\mu =\frac{1}{2}\mu _1+\frac{1}{2}\mu _2$ which contradicts the fact that
$\mu\in\rmext\mscript (\ascript _4)$. Since this method applies to all the cases, we conclude that
$\mscript (\ascript _4)\subseteq\pscript (\ascript _4)$.
\end{proof}

\begin{cor}       
\label{cor42}
If $\mu$ is a $q$-measure on $\ascript _n$, then $\mu$ has the form
$\mu =\sum _{i=1}^m\lambda _i\mu _i$ for $\lambda _i>0$ and $\mu _i\in\pscript (\ascript _n)$.
\end{cor}
\begin{proof}
Since $\mscript (\ascript _n)$ is a compact convex subset of the finite dimensional real linear space $\sscript (\ascript _n)$ by the Krein-Milman theorem, $\mscript (\ascript _n)$ is the closed convex hull of its extremal elements. Applying Theorem~\ref{thm41}, we have that
$\rmext\mscript (\ascript _n)=\pscript (\ascript _n)$. Since $\pscript (\ascript _n)$ is finite, the convex hull of $\pscript (\ascript _n)$ is already closed. Hence, every $\mu\in\mscript (\ascript _n)$ has the form $\mu =\sum _{i=1}^m\lambda _i\mu _i$, where $\mu _i\in\pscript (\ascript _n)$,
$\lambda _i>0$ and $\sum _{i=1}^m\lambda _i=1$. If $\mu$ is a $q$-measure, letting
\begin{equation*}
M=\max\brac{\mu (A)\colon A\in\ascript _n}
\end{equation*}
we have that $\mu /M\in\mscript (\ascript _n)$. The result now follows.
\end{proof}

A coevent $\phi\in\ascript _n^*$ can be considered as a map $\phi\colon\ascript _n\to\brac{0,1}$ where we view $\brac{0,1}\subseteq\real$ with the usual addition and multiplication. If it happens that $\phi\in\mscript (\ascript _n)$ or equivalently $\phi\in\pscript (\ascript _n)$, then we say that
$\phi$ is a \textit{pure coevent}. Conversely, if $\mu$ is a pure $q$-measure, then we call the map
$\muhat\colon\ascript _n\to\integers _2$ with the same values as $\mu$ the \textit{corresponding} pure coevent. The set of all pure coevents in $\ascript _n^*$ is denoted by $\ascript _{n,p}^*$ and is called the \textit{pure anhomomorphic logic}. It is clear that every coevent in $\ascript _2^*$ is pure so that $\ascript _{2,p}^*=\ascript _2^*$. However, there are only 34 pure coevents in $\ascript _3^*$ out of a total $2^{2^3-1}=128$ coevents \cite{gud092}.\medskip

\begin{exam}{7}    
Examples of pure coevents in $\ascript _3^*$ are $\omega _1^*$,
$\omega _1^*\oplus\omega _2^*$, $\omega _1^*\oplus\omega _1^*\omega _2^*$
$\omega _1^*\omega _2^*$, $\omega _1^*\oplus\omega _2^*\oplus\omega _1^*\omega _2^*$,
$\omega _1^*\oplus\omega _1^*\omega _2^*\oplus\omega _2^*\omega _3^*$,
$\omega _1^*\oplus\omega _2^*\oplus\omega _1^*\omega _2^*\oplus\omega _1^*\omega _3^*$,
$\omega _1^*\oplus\omega _1^*\omega _2^*\oplus\omega _1^*\omega _3^*
  \oplus\omega _2^*\omega _3^*$,
$\omega _1^*\oplus\omega _2^*\oplus\omega _3^*\oplus\omega _1^*\omega _2^*
  \oplus\omega _1^*\omega _3^*\oplus\omega _2^*\omega _3^*$
and the rest are obtained by symmetry. An example of a $\phi\in\ascript _3^*$ that is not pure is
$\phi =\omega _1^*\oplus\omega _2^*\oplus\omega _3^*$. Indeed, $\phi (\Omega _3)=1$ and
\begin{equation*}
\phi\paren{\brac{\omega _1,\omega _2}}+\phi\paren{\brac{\omega _1,\omega _3}}
  +\phi\paren{\brac{\omega _2,\omega _3}}-\phi (\omega _1)-\phi (\omega _2)-\phi (\omega _3)=-3
\end{equation*}
Another example of a nonpure element of $\ascript _3^*$ is
$\psi =\omega _1^*\oplus\omega _2^*\omega _3^*$. Indeed $\psi (\Omega _3)=0$ and
\begin{equation*}
\psi\paren{\brac{\omega _1,\omega _2}}+\psi\paren{\brac{\omega _1,\omega _3}}
  +\psi\paren{\brac{\omega _2,\omega _3}}-\psi (\omega _1)-\psi (\omega _2)-\psi (\omega _3)=2
\end{equation*}
\end{exam}

\begin{lem}       
\label{lem43}
If $\phi\in\ascript _{n,p}^*$, then $\phi$ is quadratic.
\end{lem}
\begin{proof}
We must show that if $\phi$ satisfies
\begin{equation}         
\label{eq47}
\phi (A\cupdot B\cupdot C)=\phi (A\cupdot B)+\phi (A\cupdot C)+\phi (B\cupdot C)
  -\phi (A)-\phi (B)-\phi (C)
\end{equation}
then $\phi$ satisfies
\begin{equation}         
\label{eq48}
\phi (A\cupdot B\cupdot C)=\phi (A\cupdot B)\oplus\phi (A\cupdot C)\oplus\phi (B\cupdot C)
  \oplus\phi (A)\oplus\phi (B)\oplus\phi (C)
\end{equation}
Suppose the left hand side of \eqref{eq48} is $1$. Then there are an odd number of $1$s on the right hand side of \eqref{eq47}. Hence, the right hand side of \eqref{eq48} is $1$. Suppose the left hand side of \eqref{eq48} is $0$. Then there are an even number of $1$s on the right hand side of \eqref{eq47}. Hence, the right hand side of \eqref{eq48} is $0$. We conclude that \eqref{eq48} holds so $\phi$ is quadratic.
\end{proof}

It follows from Lemma~\ref{lem43} that $\ascript _{n,p}^*\subseteq\ascript _{n,q}^*$. Of course, not all quadratic coevents are pure. For instance, in Example~7 we showed that the quadratic coevents $\omega _1^*\oplus\omega _2^*\oplus\omega _3^*$ and
$\omega _1^*\oplus\omega _2^*\omega _3^*$ are not pure. Also, there are 64 quadratic coevents in $\ascript _3^*$ and only 34 pure coevents.

If $\ascript _{n,0}^*\subseteq\ascript _n^*$, we can place a measure $\nu$ on $\ascript _{n,0}^*$ by specifying $\nu\paren{\brac{\phi}}\ge 0$ on $\brac{\phi}$ for every
$\phi\in\ascript _{n,0}^*$ and extending $\mu$ to the power set $P(\ascript _{n,0}^*)$ of
$\ascript _{n,0}^*$ by additivity. A $q$-measure $\mu$ on $\ascript _n$ \textit{transfers} to
$\ascript _{n,0}^*$ if there exists a measure $\nu$ on $P(\ascript _{n,0}^*)$ such that
\begin{equation}         
\label{eq49}
\nu\paren{\brac{\phi\in\ascript _{n,0}^*\colon\phi (A)=1}}=\mu (A)
\end{equation}
for all $A\in\ascript _n$ \cite{gtw092}. If $\mu$ transfers to $\nu$, then the quantum dynamics given by $\mu$ can be described by a classical dynamics given by $\nu$. Moreover, if $\mu (A)=0$ then
$\brac{\phi\in\ascript _{n,0}\colon\phi (A)=1}$ has $\nu$ measure zero which is a form of preclusivity \cite{sor941, sor942}.

\begin{thm}       
\label{thm44}
If $\ascript _{n,p}^*\subseteq\ascript _{n,0}^*\subseteq\ascript _n^*$, then any $q$-measure $\mu$ on $\ascript _n$ transfers to $\ascript _{n,0}^*$.
\end{thm}
\begin{proof}
By Corollary ~\ref{cor42}, $\mu$ has the form $\mu =\sum _{i=1}^m\lambda _i\mu _i$ for
$\lambda _i>0$ and $\mu _i\in\pscript (\ascript _n)$. Let $\nu$ be the measure on
$P(\ascript _{n,0}^*)$ defined by
\begin{equation}         
\label{eq410}
\nu =\sum _{i=1}^m\lambda _i\delta _{\muhat _i}
\end{equation}
Then for every $A\in\ascript _n$ we have that
\begin{align*}
\nu\paren{\brac{\phi\in\ascript _{n,0}^*\colon\phi (A)=1}} &=\sum _{i=1}^m\lambda _i
  \delta _{\muhat _i}\paren{\brac{\phi\in\ascript _{n,0}^*\colon\phi (A)=1}}\\
  &=\sum _{i=1}^m\lambda _i\mu _i(A)=\mu (A)
\qedhere
\end{align*}
\end{proof}

It follows from Theorem~\ref{thm44} that any $q$-measure on $\ascript _n$ transfers to
$\ascript _{n,p}^*$, $\ascript _{n,q}^*$ and $\ascript _n^*$.

\begin{cor}       
\label{cor45}
If $\nu$  is a measure on $P(\ascript _{n,p}^*)$, the there exists a $q$-measure $\mu$ on
$\ascript _n$ that transfers to $\nu$; that is, $\mu$ satisfies \eqref{eq49}.
\end{cor}
\begin{proof}
We can represent $\nu$ as in \eqref{eq410} where $\lambda _i=\nu (\muhat _i)$,
$i=1,\ldots ,m$. If $\mu$ is the $q$-measure on $\ascript _n$ given by
$\mu =\sum _{i=1}^m\lambda _i\mu _i$, then by the proof of Theorem~\ref{thm44} we have that
$\mu$ satisfies \eqref{eq49}.
\end{proof}

The next Corollary shows what happens if the $q$-measure $\mu$ turns out to be a measure.

\begin{cor}       
\label{cor46 }
Suppose $\ascript _{n,0}^*$ satisfies
$\ascript _{n,a}^*\subseteq\ascript _{n,0}^*\subseteq\ascript _n^*$. If $\mu$ is a measure on $\ascript _n$, then $\mu$ transfers to a measure $\nu$ on $\ascript _{n,0}^*$ satisfying $\nu (\phi )=0$ unless $\phi$ is additive. Conversely, if a measure $\nu$ on $\ascript _{n,0}^*$ satisfies
$\mu (\phi )=0$ unless $\phi$ is additive, then there is a measure on $\ascript _n$ that transfers to
$\nu$.
\end{cor}
\begin{proof}
If $\mu$ is a measure on $\ascript _n$ then $\mu$ has the form
\begin{equation}         
\label{eq411}
\mu =\sum _{i=1}^n\lambda _i\delta _{\omega _i}
\end{equation}
where $\lambda _i\ge 0$, $i=1,\ldots ,n$. Hence, $\mu$ transfers to 
\begin{equation}         
\label{eq412}
\nu =\sum _{i=1}^n\lambda _i\delta _{\omega _i^*}
\end{equation}
on $\ascript _{n,0}^*$ where $\nu (\phi )=0$ unless $\phi$ is additive. Conversely, if $\nu$ is a measure on $\ascript _{n,0}^*$ satisfying $\nu (\phi )=0$ unless $\phi$ is additive, then $\nu$ has the form \eqref{eq412}. Hence, the measure $\mu$ given by \eqref{eq411} transfers to $\nu$.
\end{proof}

Motivated by Theorem~\ref{thm44} one might conjecture that if every $q$-measure on $\ascript _n$ transfers to $\ascript _{n,0}^*$, then $\ascript _{n,p}^*\subseteq\ascript _{n,0}^*$; that is,
$\ascript _{n,p}^*$ is the smallest subset of $\ascript _n^*$ to which every $q$-measure on
$\ascript _n$ transfers. The next example shows that this conjecture does not hold.

\begin{exam}{8}    
We have that
\begin{equation*}
\ascript _{2,p}^*=\ascript _2^*=\brac{0,\omega _1^*,\omega _2^*,\omega _1^*\oplus\omega _2^*,
  \omega _1^*\omega _2^*,\omega _1^*\oplus\omega _1^*\omega _2^*,
  \omega _2^*\oplus\omega _1^*\omega _2^*,1}
\end{equation*}
Let $\phi _1=\omega _1^*$, $\phi _2=\omega _2^*$, $\phi _3=\omega _1^*\oplus\omega _2^*$,
$\phi _4=\omega _1^*\omega _2^*$, $\phi _5=\omega _1^*\oplus\omega _1\omega _2^*$ and
$\phi _6=\omega _2^*\oplus\omega _1^*\omega _2^*$. If $\mu$ is an arbitrary $q$-measure on
$\ascript _2$, then $\mu$ transfers to a measure $\nu$ on $\ascript _2^*$ and we have that 
\begin{align*}
\mu (\omega _1)&=\nu\paren{\brac{\phi\in\ascript _2^*\colon\phi (\omega _1)=1}}
  =\nu (\phi _1)+\nu (\phi _3)+\nu (\phi _5)+\nu (1)\\
\mu (\omega _2)&=\nu\paren{\brac{\phi\in\ascript _2^*\colon\phi (\omega _2)=1}}
  =\nu (\phi _2)+\nu (\phi _3)+\nu (\phi _6)+\nu (1)\\
\mu (\Omega _2)&=\nu\paren{\brac{\phi\in\ascript _2^*\colon\phi (\Omega _2)=1}}
  =\nu (\phi _1)+\nu (\phi _2)+\nu (\phi _4)+\nu (1)
\end{align*}
Letting $\ascript _{2,0}^*=\brac{\phi _4,\phi _5,\phi _6}$, we have that every $q$-measure on
$\ascript _2$ transfers to $\ascript _{2,0}^*$. In fact, if $\mu$ is a $q$-measure on $\ascript _2$, then $\mu$ transfers to the measure $\nu$ on $\ascript _{2,0}^*$ given by
$\nu (\phi _4)=\mu (\Omega _2)$, $\nu (\phi _5)=\mu (\omega _1)$,
$\nu (\phi _6)=\mu (\omega _2)$. This example also shows that the measure that $\mu$ transfers to need not be unique. For instance, let $\mu$ be the $q$-measure on $\ascript _2$ given by
$\mu (\omega _1)=1$, $\mu (\omega _2)=1$, $\mu (\Omega _2)=0$. Then $\mu$ transfers to
$\nu _1$ on $\ascript _2^*$ given by $\nu _1(\phi _5)=\nu _1(\phi _6)=1$ and $\nu (\phi )=0$,
$\phi\ne\phi _5,\phi _6$. Also, $\mu$ transfers to $\nu _2$ on $\ascript _2^*$ given by
$\nu _2(\phi _3)=1$ and $\nu (\phi )=0$ for $\phi\ne\phi _3$.
\end{exam}

\begin{exam}{9}    
We use the same notation as in Example~8. We first show that a $q$-measure $\mu$ on
$\ascript _2$ transfers to
\begin{equation*}
\ascript _{2,m}^*=\brac{0,\phi _1,\phi _2,\phi _4}
\end{equation*}
if and only if $\mu (\Omega _2)\ge\mu (\omega _1)+\mu (\omega _2)$. If $\mu$ transfers to $\nu$, then $\mu (\omega _1)=\nu (\phi _1)$, $\mu (\omega _2)=\nu (\phi _2)$ and
\begin{equation*}
\mu (\Omega _2)=\nu (\phi _1)+\nu (\phi _2)+\nu (\phi _4)
\end{equation*}
Hence, $\mu (\Omega _2)\ge\mu (\omega _1+\mu (\omega _2)$. Conversely, if
$\mu (\Omega _2)\ge\mu (\omega _1)+\mu (\omega _2)$ then letting
$\nu (\phi _1)=\mu (\omega _1)$, $\nu (\phi _2)=\mu (\omega _2)$, $\nu (0)=0$ and 
\begin{equation*}
\nu (\phi _4)=\mu (\Omega _2)-\mu (\omega _1)-\mu (\omega _2)
\end{equation*}
we see that $\mu$ transfers to $\nu$ on $\ascript _{2,m}^*$. We next show that a $q$-measure
$\mu$ on $\ascript _2$ transfers to
\begin{equation*}
\ascript _{2,a}^*=\brac{0,\phi _1,\phi _2,\phi _3}
\end{equation*}
if and only if
\begin{equation*}
\ab{\mu (\omega _1)-\mu (\omega _2)}\le\mu (\Omega _2)\le\mu (\omega _1)+\mu (\omega _2)
\end{equation*}
If $\mu$ transfers to $\nu$, then $\mu (\omega _1)=\nu (\phi _1)+\nu (\phi _3)$,
$\mu (\omega _2)=\nu (\phi _2)+\nu (\phi _3)$ and $\mu (\Omega _2)=\nu (\phi _1)+\nu (\phi _2)$. Hence, $\mu (\Omega _2)\le\mu (\omega _1)+\mu (\omega _2)$ and
\begin{align*}
\mu (\omega _1)-\mu (\omega _2)&=\nu (\phi _1)+\nu (\phi _2)\le\mu (\Omega _2)\\
\mu (\omega _2)-\mu (\omega _1)&=\nu (\phi _2)-\nu (\phi _1)\le\mu (\Omega _2)
\end{align*}
so the given inequalities hold. Conversely, if the inequalities hold, then letting $\nu (0)=0$ and
\begin{align*}
\nu (\phi _1)&=\tfrac{1}{2}\sqbrac{\mu (\Omega _2)+\mu (\omega _1)-\mu (\omega _2)}\\
\nu (\phi _2)&=\tfrac{1}{2}\sqbrac{\mu (\Omega _2)-\mu (\omega _1)+\mu (\omega _2)}\\
\nu (\phi _3)&=\tfrac{1}{2}\sqbrac{\mu (\omega _1)+\mu (\omega _2)-\mu (\Omega _2)}
\end{align*}
we see that $\mu$ transfers to $\nu$ on $\ascript _{2,a}^*$.
\end{exam}

The next theorem further shows that the multiplicative anhomomorphic logic $\ascript _{n,m}^*$ is not adequate for transferring $q$-measures. This result was proved in \cite{gt09, gtw092}. However, our proof is simpler and more direct.

\begin{thm}       
\label{thm47}
If a $q$-measure $\mu$ on $\ascript _n$ transfers to a measure $\nu$ on $\ascript _{n,m}^*$, then
$\nu (\phi )=0$ unless $\phi$ is quadratic.
\end{thm}
\begin{proof}
We have that
\begin{equation*}
\ascript _{n,m}^*=\brac{0,\omega _1^*,\ldots ,\omega _n^*,\omega _1^*\omega _2^*,
  \ldots ,\omega _{n-1}^*\omega _n^*,\omega _1^*\omega _2^*\omega _3^*
  \ldots ,\omega _1^*\omega _2^*\cdots\omega _n^*}
\end{equation*}
Since $\mu (A)=\nu\paren{\brac{\phi\in\ascript _{n,m}^*\colon\phi (A)=1}}$ we have that
$\mu (\omega _i)=\nu (\omega _i^*)$, $i=1,\ldots ,n$ and
\begin{equation*}
\mu\paren{\brac{\omega _i,\omega _j}}
  =\nu (\omega _i^*)+\nu (\omega _j^*)+\nu (\omega _i^*\omega _j^*)
\end{equation*}
$i,j=1,\ldots ,n$. Now
\begin{align}         
\label{eq413}
\mu\paren{\brac{\omega _1,\omega _2,\omega _3}}
  &=\nu (\omega _1^*)+\nu (\omega _2^*)+\nu (\omega _3^*)+\nu (\omega _1^*\omega _2^*)
  +\nu (\omega _1^*\omega _3^*)\notag\\
  &\quad +\nu (\omega _2^*\omega _3^*)+\nu (\omega _1^*\omega _2^*\omega _3^*)
\end{align}
Since $\mu$ is grade-2 additive we have that
\begin{align}         
\label{eq414}
\mu\paren{\brac{\omega _1,\omega _2,\omega _3}}
  &=\sum _{i<j=1}^3\mu\paren{\brac{\omega _i,\omega _j}}-\sum _{i=1}^3\mu (\omega _i)
  \notag\\
  &=\sum _{i<j=1}^3\nu (\omega _i^*\omega _j^*)+\sum _{i=1}^n\nu (\omega _i^*)
\end{align}
Comparing \eqref{eq413} and \eqref{eq414} shows that
$\nu (\omega _1^*\omega _2^*\omega _3^*)=0$. In a similar way, we conclude that
$\nu (\omega _i^*\omega _j^*\omega _k^*)=0$, $i,j,k=1,\ldots ,n$, $i<j<k$. Next,
\begin{align}         
\label{eq415}
\mu\paren{\brac{\omega _1,\omega _2,\omega _3,\omega _4}}
  &=\nu (\omega _1^*)+\cdots +\nu (\omega _4^*)+\nu (\omega _1^*\omega _2^*)\notag\\
  &\quad +\cdots +\nu (\omega _3^*\omega _4^*)
  +\nu (\omega _1^*\omega _2^*\omega _3^*\omega _4^*)
\end{align}
Since $\mu$ is grade-2 additive we have that
\begin{align}         
\label{eq416}
\mu\paren{\brac{\omega _1,\omega _2,\omega _3,\omega _4}}
  &=\sum _{i<j=1}^4\mu\paren{\brac{\omega _i,\omega _j}}-2\sum _{i=1}^4\mu (\omega _i)
  \notag\\
  &=\sum _{i<j=1}^4\nu (\omega _i^*\omega _j^*)+\sum _{i=1}^4\nu (\omega _i^*)
\end{align}
Comparing \eqref{eq415} and \eqref{eq416} shows that
$\nu (\omega _1^*\omega  _2^*\omega _3^*\omega _4^*)=0$. In a similar way, we conclude that
$\nu (\omega _i^*\omega _j^*\omega _k^*\omega _l^*)=0$, $i,j,k,l=1,\ldots ,n$, $i<j<k<l$. Continuing by induction, we have that $\nu (\phi )=0$ unless $\phi$ is quadratic.
\end{proof}

The result (and proof) in Theorem~\ref{thm47} holds if $\ascript _{n,m}^*$ is replaced by any subset $\ascript _{n,0}^*\subseteq\ascript _{n,m}^*$. Theorem~\ref{thm47} also shows that for
$\ascript _{n,0}^*\subseteq\ascript _{n,m}^*$, if a $q$-measure $\mu$ transfers from $\ascript _n$ to $\ascript _{n,0}^*$ then 
\begin{equation*}
\mu\paren{\brac{\omega _i,\omega _j}}\ge\mu (\omega _i)+\mu (\omega _j)
\end{equation*}
$i,j=1,\ldots ,n$, $i\ne j$. Hence, not all $q$-measures can be transferred to $\ascript _{n,m}^*$.


\begin{thebibliography}{99}
\bibitem{dgt08}F.~Dowker and Y.~Ghazi-Tabatabai, 
Dynamical wave function collapse models in quantum measure theory, \textit{J. Phys. A} \textbf{41}, 205306  (2008).
\bibitem{gt09}Y.~Ghazi-Tabatabai, 
Quantum measure theory: a new interpretation, arXiv: quant-ph (0906.0294), 2009.
\bibitem{gtw091}Y.~Ghazi-Tabatabai and P.~Wallden, The emergence of probabilities in anhomomorphic logic,
\textit{J. Phys. A} \textbf{42}, 235303 (2009).
\bibitem{gtw092}Y.~Ghazi-Tabatabai and P.~Wallden, Dynamics and predictions in the co-event interpretation, arXiv: quant-ph (0901.3675), 2009.
\bibitem{gudamm}S.~Gudder, Finite quantum measure spaces, \textit{Amer. Math. Monthly} (2010).
\bibitem{gudms}S.~Gudder, Quantum measure theory, \textit{Math. Slovaca} (to appear).
\bibitem{gud091}S.~Gudder, Quantum measure and integration theory, arXiv: quant-ph (0909.2203), 2009 and \textit{J. Math. Phys.} (to appear).
\bibitem{gud092}S.~Gudder, An anhomomorphic logic for quantum mechanics, arXiv: quant-ph (0910.3253), 2009.
\bibitem{gud093}S.~Gudder, Quantum integrals and anhomomorphic logics, arXiv: quant-ph (0911.1572), 2009.
\bibitem{gud10}S.~Gudder, Quantum reality filters, arXiv: quant-ph (1002.4225), 2010.
\bibitem{hal09}J.~J.~Halliwell, Partial decoherence of histories and the Diosi test, arXiv: quant-ph (0904.4388), 2009 and \textit{Quantum Information Processing} (to appear).
\bibitem{mocs05}X.~Martin, D.~O'Connor and R.~Sorkin, The random walk in generalized quantum theory, \textit{Phys. Rev.} \textbf{071}, 024029 (2005).
\bibitem{sal02}R.~Salgado, Some identities for the quantum measure and its generalizations,
\textit{Mod. Phys.  Letts.~A} \textbf{17} (2002), 711--728.
\bibitem{sor941}R.~Sorkin, 
Quantum mechanics as quantum measure theory, \textit{Mod. Phys. Letts.~A} \textbf{9} (1994), 3119--3127.
\bibitem{sor942}R.~Sorkin, 
Quantum measure theory and its interpretation, in \textit{Proceedings of the 4th Drexel Symposium on quantum nonintegrability}, eds. D.~Feng and B.-L.~Hu, 229--251 (1994).
\bibitem{sor07}R.~Sorkin, 
Quantum mechanics without the wave function, \textit{J.~Phys.~A} \textbf{40} (2007), 3207--3231.
\bibitem{sor09}R.~Sorkin,  
An exercise in ``anhomomorphic logic,'' \textit{J. Phys.} (to appear).
\bibitem{sw08}S.~Surya and P.~Wallden, 
Quantum covers in quantum measure theory, arXiv, quant-ph (0809.1951), 2008.

\end{thebibliography}
\end{document}